\documentclass[final,1p,times,number]{elsarticle}
\usepackage{amsmath}
\usepackage{amsthm}
\usepackage{amssymb}
\usepackage{amsfonts}
\usepackage{bm}
\usepackage{tikz}
\usetikzlibrary{shapes,arrows}
\usepackage{float}
\usepackage{color}
\usepackage{array}
\newcolumntype{C}[1]{>{\centering}m{#1}}
\usepackage{todonotes}
\usepackage{footmisc}
\usepackage{url}
\usepackage{enumitem}
\usepackage{geometry}
\usepackage[ruled, vlined, linesnumbered]{algorithm2e}

\newtheorem{theorem}{Theorem}
\newtheorem{lemma}[theorem]{Lemma}
\newtheorem{corollary}[theorem]{Corollary}

\newtheorem{proposition}[theorem]{Proposition}
\newtheorem{definition}[theorem]{Definition}

\newtheorem{example}[theorem]{Example}
\newtheorem{remarks}[theorem]{Remark}

\DeclareMathOperator{\ini}{ini}

\DeclareMathOperator{\lt}{lt}
\DeclareMathOperator{\lv}{lv}

\DeclareMathOperator{\prem}{prem}
\DeclareMathOperator{\sat}{sat}
\DeclareMathOperator{\wcs}{wcs}

\newcommand{\field}[1]{\mathbb{#1}}
\newcommand{\fk}{\field{K}}
\newcommand{\pset}[1]{\mathcal{#1}}
\newcommand{\ideal}[1]{\mathfrak{#1}}
\newcommand{\bases}[1]{\langle #1 \rangle}
\newcommand{\rad}[1]{\sqrt{#1}}
\DeclareMathOperator{\zero}{\mathsf{Z}}

\newcommand{\src}{src}
\newcommand{\grobner}{Gr\"{o}bner }
\newcommand{\p}[1]{\bm{#1}}
\newcommand{\kx}{\field{K}[\p{x}]}
\DeclareMathOperator{\nform}{nform}

\DeclareMathOperator{\algnorr}{\sf srcDec}

\DeclareMathOperator{\algnot}{\sf srcDivisor}
\DeclareMathOperator{\algt}{\sf srcPair}

\newcommand{\point}[1]{\bm{#1}}
\begin{document}

\begin{frontmatter}

\title{Computing strong regular characteristic pairs\\ with \grobner bases\tnoteref{t1}}
\tnotetext[t1]{This work was partially supported by the National Natural Science Foundation of China (NSFC 11771034).}

\author[1]{Rina Dong\corref{cor1}}
\ead{rina.dong@buaa.edu.cn}
\author[1,2]{Dongming Wang}
\ead{dongming.wang@lip6.fr}
\cortext[cor1]{Corresponding author: rina.dong@buaa.edu.cn, School of Mathematical Sciences, Beihang University, Beijing 100191, China.}
\address[1]{Beijing Advanced Innovation Center for Big Data and Brain Computing -- LMIB -- School of Mathematical Sciences,
	Beihang University, Beijing 100191, China}
\address[2]{Centre National de la Recherche Scientifique,
	75794 Paris cedex 16, France}

%

\begin{abstract}
The W-characteristic set of a polynomial ideal is the minimal triangular set contained in the reduced lexicographical \grobner basis of the ideal. A pair $(\pset{G},\pset{C})$ of polynomial sets is a strong regular characteristic pair if $\pset{G}$ is a reduced lexicographical \grobner basis, $\pset{C}$ is the W-characteristic set of the ideal $\bases{\pset{G}}$, the saturated ideal $\sat(\pset{C})$ of $\pset{C}$ is equal to $\bases{\pset{G}}$, and $\pset{C}$ is regular. In this paper, we show that
for any polynomial ideal $\ideal{I}$ with given generators one can either detect that $\ideal{I}$ is unit, or construct a strong regular characteristic pair $(\pset{G},\pset{C})$ by computing \grobner bases such that $\ideal{I}\subseteq \sat(\pset{C})=\bases{\pset{G}}$ and $\sat(\pset{C})$ divides $\ideal{I}$, so the ideal $\ideal{I}$ can be split into the saturated ideal $\sat(\pset{C})$ and the quotient ideal $\ideal{I}:\sat(\pset{C})$. Based on this strategy of splitting by means of quotient and with \grobner basis and ideal computations, we devise a simple algorithm to decompose an arbitrary polynomial set $\pset{F}$ into finitely many strong regular characteristic pairs, from which two representations for the zeros of $\pset{F}$ are obtained: one in terms of strong regular \grobner bases and the other in terms of regular triangular sets. We present some properties about strong regular characteristic pairs and characteristic decomposition and illustrate the proposed algorithm and its performance by examples and experimental results.
\end{abstract}

\end{frontmatter}

\section{Introduction}
\label{sec:intro}
Triangular sets \cite{r50d,w86b,k93g,W2001E} and \grobner bases \cite{B1965A,B85G,B93G,CLO1997I} are special kinds of well-structured sets of multivariate polynomials that can be used to represent and to study zeros of arbitrary polynomial sets and ideals. A large variety of problems in commutative algebra and algebraic geometry \cite{CLO1997I,cox98} may readily be solved by transforming the involved sets of polynomials into triangular sets or \grobner bases. It is widely known that the theories and methods of triangular sets are different from those of \grobner bases conceptually and operationally. The questions that have motivated our work here and in \cite{w2016o,WDM20} are what inherent relationship there may exist between triangular sets and \grobner bases and how to connect or combine the two algorithmic approaches to amplify their applicability and power. These questions have been touched in \cite{L85i,MM2003r,D2012o,DL92} for bivariate, zero-dimensional, and other kinds of special polynomial ideals. For general polynomial ideals of arbitrary dimension, it is shown in \cite{a99t,w2016o} that intrinsical connections between Ritt characteristic sets and lexicographical \grobner bases exist and the concept of W-characteristic sets plays an essential role in exploiting such connections. Apart from the lack of investigation on the connection aspect, the literature on triangular sets and \grobner bases is extremely rich (see \cite{a99t,B93G,B85G,Rob12,c07c,CM2012a,D2012o,FGLM93E,GVW2016,Gao1993On,k93g,L85i,li08a,LI2010D,MON02,m00t,MW19,M2012d,SY1996l,w00c,w2016o,WDM20} and references therein).

Multivariate polynomials in a triangular set may be ordered strictly according to their leading variables, with respect to a fixed variable ordering, so the number of the polynomials cannot be bigger than that of the variables in any triangular set. On the other hand, \grobner bases are defined with respect to a fixed term order determined by the variable ordering and the number of elements in a \grobner basis can be arbitrarily large. For any polynomial ideal with given set $\pset{F}$ of generators, one can compute, by using any of the available algorithms, a \grobner basis that generates the same ideal as $\pset{F}$. To represent the zeros of $\pset{F}$ using triangular sets, in general one needs more than one triangular set, so decomposition takes place. When a triangular set $\pset{T}$ is of concern, the leading coefficients of the polynomials in $\pset{T}$ with respect to their leading variables, called the initials of the polynomials in $\pset{T}$, play a fundamental role. The saturated ideal of $\pset{T}$ by the product of the initials of the polynomials in $\pset{T}$ is simply called the saturated ideal of $\pset{T}$ and denoted as $\sat(\pset{T})$. Its radical is the largest ideal whose zero set contains the set of those zeros of $\pset{T}$ which are not zeros of any of the initials.

Triangular sets may be ordered according to the ranks (leading variables and degrees) and then the leading terms of their polynomials. Let the minimal triangular set contained in the reduced lexicographical (lex) \grobner basis of a polynomial ideal (or trivially $[1]$ if the ideal is unit) be called the W-characteristic set of the ideal. By strong regular \grobner basis, we mean a reduced lex \grobner basis $\pset{G}$ such that the W-characteristic set $\pset{C}$ of the ideal $\bases{\pset{G}}$ is regular and $\sat(\pset{C})=\bases{\pset{G}}$; we call $(\pset{G},\pset{C})$ a strong regular characteristic pair, or an \src~pair for short. The pair $(\pset{G},\pset{C})$ is an interesting object of study because the strong regular \grobner basis $\pset{G}$ and the regular triangular set $\pset{C}$ therein not only have remarkable properties but also provide two different yet correlated representations for the zeros of the ideal $\bases{\pset{G}}$.

What interests us most is algorithmic decomposition of arbitrary polynomial sets into strong regular \grobner bases, or equivalently into \src~pairs, for which we have the following general approach. From any polynomial set $\pset{F}$, one can compute finitely many regular sets $\pset{T}_1,\ldots,\pset{T}_e$ (also called regular chains; \cite{K1991M}) such that
\begin{equation}\label{eq:kalk}
	\rad{\bases{\pset{F}}}= \rad{\sat(\pset{T}_1)}\cap \cdots\cap \rad{\sat(\pset{T}_e)}.
\end{equation}
There are two families of algorithms for such regular triangular decomposition. One family of algorithms was proposed initially by Kalkbrener \cite{k93g}, and developed further by Moreno Maza and coauthors \cite{m00t,ACM2014d,c07c,CM2012a} with algorithmic techniques from the method of Lazard \cite{l91n}. These algorithms were designed mainly for computing regular triangular representations of the form \eqref{eq:kalk}. The other family of algorithms was proposed by the second author \cite{W2001E,w00c} to compute regular zero decompositions of the form
\begin{equation}\label{eq:wang}
	\zero(\pset{F})=\zero(\pset{T}_1/ J_1)\cup \cdots \cup \zero(\pset{T}_e/ J_e), 	
\end{equation}
where each $J_i$ is the product of the initials of the polynomials in the regular set $\pset{T}_i$ and $\zero(\pset{T}_i/J_i)$ denotes the set of all common zeros of the polynomials in $\pset{T}_i$ which do not make $J_i$ vanish.
It is easy to see that decomposition \eqref{eq:wang} implies representation \eqref{eq:kalk}. However, representation \eqref{eq:kalk} does not necessarily lead to decomposition \eqref{eq:wang} and the generators of the saturated ideals $\sat(\pset{T}_i)$ in \eqref{eq:kalk} are not explicitly provided. Nevertheless, for each $\sat(\pset{T}_i)$ a \grobner basis can be computed straightforwardly from $\pset{T}_i$.

Another alternative approach proposed recently by Mou and the authors \cite{WDM20} permits one to decompose an arbitrary polynomial set into normal or regular characteristic pairs directly. This approach is based on a structure theorem about irregular W-characteristic sets \cite{w2016o} for the splitting of ideals and relies strongly upon \grobner basis computation; it is independent of pseudo-division-based triangular decomposition.

In this paper, we show that
from any polynomial ideal $\ideal{I}$ with given generators one can either detect that $\ideal{I}$ is unit, or construct an \src~pair $(\pset{G},\pset{C})$ such that $\ideal{I}\subseteq \sat(\pset{C})$ and $\sat(\pset{C})$ divides $\ideal{I}$ by computing \grobner bases. After $(\pset{G},\pset{C})$ (called an \src~divisor of $\ideal{I}$) is constructed, one can divide the saturated ideal $\sat(\pset{C})$ out of $\ideal{I}$ to obtain an ideal $\ideal{J}$
by taking ideal quotient: $\ideal{J}=\ideal{I}:\sat(\pset{C})$. Then the radical ideal of $\ideal{I}$ is decomposed as the intersection of the radical ideals of $\sat(\pset{C})$ and $\ideal{J}$. Two key techniques are used in search for an \src~divisor of $\ideal{I}$: one is to produce an \src~pair by means of computing $\sat(\wcs(\cdots\sat(\wcs(\ideal{I}))))$, where $\wcs(\ideal{I})$ denotes the W-characteristic set of $\ideal{I}$. The other is to use the relation $\sat(\pset{C}^*)\neq\ideal{I}$, when the produced \src~pair is not an \src~divisor of $\ideal{I}$,
to determine a polynomial $F\in\sat(\pset{C}^*)\setminus\ideal{I}$ with $\pset{C}^*=\wcs(\ideal{I})$. In this case,
there exists a positive integer $q$ such that $H=F[\prod_{C\in \pset{C}^*}\ini(C)]^q\in\bases{\pset{C}^*}\subseteq\ideal{I}$ and an \src~divisor of $\ideal{I}$ may be found among the \src~pairs computed from $\bases{\ideal{I}\cup\{H_i\}}$ for $H_i\not\in\ideal{I}$ and $H_i\mid H$. Once $\ideal{J}$ is obtained, one can iterate the find-and-divide process with $\ideal{J}$ instead of $\ideal{I}$.
Based on this strategy of splitting by means of quotient and with \grobner basis and ideal computations, we devise a simple algorithm to decompose an arbitrary polynomial set $\pset{F}$ into finitely many \src~pairs $(\pset{G}_1,\pset{C}_1),\ldots,(\pset{G}_e,\pset{C}_e)$ such that
\begin{equation}\label{eq:chardec}
	\rad{\bases{\pset{F}}}=\rad{\sat(\pset{C}_1)}\cap \cdots \cap \rad{\sat(\pset{C}_e)}= \rad{\bases{\pset{G}_1}}\cap \cdots \cap \rad{\bases{\pset{G}_e}}.
\end{equation}

In the above decomposition computed by the ideal-division-based algorithm, each $\pset{G}_i$ is actually the reduced lex \grobner basis of a saturated ideal $\sat(\pset{T}_i)$, where $\pset{T}_i$ is the W-characteristic set of a certain ideal that contains $\bases{\pset{F}}$. The regular set $\pset{C}_i$ is the W-characteristic set of $\sat(\pset{T}_i)$, obtained from $\pset{G}_i$ as a by-product for free. Moreover, regular sets computed by our algorithm are normal in most cases and they are often simpler
 than the corresponding regular sets computed by pseudo-division\footnote{Most of the available methods for triangular decomposition use pseudo-division to eliminate variables. In doing pseudo-division, the dividend has to be multiplied by some power of the initial of the dividing polynomial. Repeated multiplication of the powers of initials necessarily creates extraneous factors, making the sizes of intermediate polynomials increase rapidly. It happens often that decomposition cannot continue after a few successive pseudo-divisions because polynomials in the pseudo-reminder sequence become larger and larger.} or subresultant-based algorithms, because the former are minimal triangular sets taken from lex \grobner bases. In general, the generating sets of the saturated ideals of normal triangular sets are much easier to compute than those of abnormal ones. More importantly, our algorithm generates few redundant components, so the decomposition process usually terminates in a few iterations. The effectiveness of the decomposition algorithm has been demonstrated by our experimental results.

The rest of the paper is organized as follows. After a brief introduction to \grobner bases, triangular decomposition, and characteristic decomposition in Section \ref{sec:pre}, we show in Section \ref{sec:satdec} how to construct
from an arbitrary ideal $\ideal{I}$ an \src~pair $(\pset{G},\pset{C})$ such that $\ideal{I}\subseteq \sat(\pset{C})=\bases{\pset{G}}$. In Section~\ref{sec:algsatdec}, it is explained how an \src~divisor of $\ideal{I}$ can be produced from any ideal $\ideal{I}\subsetneq\kx$ and the ideal-division-based algorithm for strong regular characteristic decomposition is described. The algorithm and its performance are illustrated by examples and experimental results with a preliminary implementation of the algorithm in Section~\ref{sec:exex}. The paper contains a summary of contributions in Section~\ref{sec:con} and some remarks on special cases in Appendix A.

\section{Preliminaries}
\label{sec:pre}
In this section we recall some basic notions which will be used in the following sections. For those notions which are not formally introduced in the paper, the reader may consult the references \cite{a99t,W2001E,B93G,CLO1997I}. \vspace{-3mm}
\subsection{Triangular sets, triangular decompositions, and \grobner bases}
Let $\fk$ be a field of characteristic 0 and $\fk[x_1, \ldots, x_n]$ be the ring of polynomials in $n$ ordered variables $x_1<\cdots<x_n$ with coefficients in $\fk$. Throughout the paper, we write $\point{x}$ for $(x_1,\ldots,x_n)$. Let $F$ be a polynomial in $\kx\setminus\fk$. With respect to the variable ordering, the greatest variable appearing in $F$ is called the \emph{leading variable} of $F$ and denoted as $\lv(F)$. Assume that $\lv(F) = x_i$; then $F$ can be written as $F = Ix_i^k + R$, where $I\in \fk[x_1, \ldots, x_{i-1}]$, $R
\in \fk[x_1, \ldots, x_i]$, and $\deg(R, x_i) < k=\deg(F, x_i)$ (the degree of $F$ in $x_i$). The polynomial
$I$ is called the \emph{initial} of $F$, denoted as $\ini(F)$. For any polynomial set $\pset{F}\subseteq \kx$, $\ini(\pset{F})$ stands for $\{\ini(F): F\in \pset{F}\}$.

\begin{definition}
	\rm A finite, nonempty, ordered set $[T_1, \ldots, T_r]$ of polynomials in $\kx\setminus\fk$ is called a \emph{triangular set} if $\lv(T_1) < \cdots < \lv(T_r)$.
\end{definition}
We denote by $\prem(P,Q)$ the pseudo-reminder of $P\in \kx$ with respect to $Q\in \kx\setminus \mathbb{K}$ in $\lv(Q)$. Let $\pset{T}=[T_1,\ldots,T_r]\subseteq \kx$ be any triangular set; the pseudo-reminder of $P$ with respect to $\pset{T}$ is defined as
$$\prem(P,\pset{T}) := \prem(\cdots\prem(\prem(P,T_r),T_{r-1}),\ldots,T_1).$$
The variables in $\{x_1,\ldots,x_n\}\setminus \{\lv(T_1),\ldots,\lv(T_r)\}$ are called the \emph{parameters} of $\pset{T}$.
For any two polynomial sets $\pset{F}, \pset{G} \subset \kx$, define
$$\zero(\mathcal{F}/\mathcal{G}) := \{\bar{\point{x}} \in \bar{\mathbb{K}}^{n}:\, F(\bar{\point{x}}) = 0,~G(\bar{\point{x}})\neq 0, \text{ for~all~} F\in \pset{F}, G \in \pset{G}\},$$
where $\bar{\fk}$ is the algebraic closure of $\fk$. Sometimes we write  $\zero(\mathcal{F}/\prod_{G\in\mathcal{G}}G)$ for $\zero(\mathcal{F}/\mathcal{G})$ and write $\zero(\pset{F})$ for $\zero(\pset{F}/\emptyset)$.

Let $\pset{F}\subseteq \kx$ be any polynomial set and denote by $\bases{\pset{F}}$ the ideal generated by $\pset{F}$ in $\kx$ and by $\rad{\bases{\pset{F}}}$ the radical of $\bases{\pset{F}}$. For any $\pset{P}\subseteq \kx$, $\bases{\pset{F}}:\bases{\pset{P}}$ denotes the ideal quotient of $\bases{\pset{F}}$ by $\bases{\pset{P}}$. The \emph{saturated ideal} of a triangular set $\pset{T} = [T_1, \ldots, T_r]$ is defined as $\sat(\pset{T}) := \{P\in\kx:\, \exists~i \text{ such that }PJ^i\in \bases{\pset{T}}\}$, where $J=\ini(T_1)\cdots \ini(T_r)$. For any $c\in \mathbb{K}\setminus\{0\}$, we consider $[c]$ also as a triangular set, which is trivial, and define $\sat([c])=\bases{1}$.


\begin{definition}\rm
	Let $\pset{T}=[T_1,\ldots,T_r]$ be any nontrivial triangular set in $\kx$. $\pset{T}$ is said to be \emph{regular}, or called a \emph{regular set} or a \emph{regular chain}, if $\ini(T_i)$ is neither zero nor a zero-divisor in $\kx /
	\sat(\pset{T}_{i-1})$ for all $i=1,\ldots,r$.
\end{definition}
Regular sets or chains \cite{w00c,a99t} are special triangular sets with nice properties which have been extensively studied. In particular, it is proved in \cite{a99t,W2001E} that a triangular set $\pset{T}$ is regular if and only if $\sat(\pset{T})=\{P\in \kx:\, \prem(P,\pset{T})=0\}$. The triangular set $\pset{T}$ is called a \emph{normal set} (or said to be \emph{normal}) if $\ini(\pset{T})$ does not involve any of the leading variables of the polynomials in $\pset{T}$. Obviously, any normal set is regular, while a regular set is not necessarily normal.
%

For a given term order $<$, the greatest term in a polynomial $F\in \kx$ with respect to $<$ is called the \emph{leading term} of $F$ and denoted as $\lt(F)$. In this paper, we are concerned only with $<_{\rm lex}$, the lex term order.

\begin{definition}\label{def:gb}\rm
	Let $\ideal{I}\subseteq \kx$ be an ideal, $<$ be a term order, and $\bases{\lt(\ideal{I})}$ stand for the ideal generated by the leading terms of all the polynomials in $\ideal{I}$. A finite
	set $\{G_1, \ldots, G_s\}\subseteq \mathfrak{I}$ is called a \emph{\grobner
		basis} of $\mathfrak{I}$ with respect to $<$ if $\bases{\lt(G_1), \ldots, \lt(G_s)} = \bases{\lt(\ideal{I})}$.
\end{definition}

Let $\pset{G}=\{G_1,\ldots,G_s\}$ be a \grobner basis of an ideal $\ideal{I}\subseteq \kx$ with respect to a fixed term order $<$. For any polynomial $F\in \kx$, there exists a unique polynomial $R\in \kx$, called the \emph{normal form} of $F$ with respect to $\pset{G}$ and denoted as $\nform(F,\pset{G})$, such that $F-R\in \ideal{I}$ and no term of $R$ is divisible by any of $\lt(G_1),\ldots,\lt(G_s)$. If $F=R$, then $F$ is said to be \emph{B-reduced} with respect to $\pset{G}$.

\begin{definition}\rm
	A \grobner basis $\{G_1,\ldots, G_s\}$ is said to be \emph{reduced} if every $G_i$ is monic and no term of $G_i$ is divisible by any $\lt(G_j)$ for all $j\neq i$ and $i,j=1,\ldots, s$.
\end{definition}

In the rest of this paper, the variable ordering will be fixed and all \grobner bases mentioned are meant reduced lex \grobner bases.

\subsection{W-characteristic sets, characteristic pairs, and characteristic  decompositions}
For any polynomial ideal, one can compute its unique reduced lex \grobner basis and from the \grobner basis, one can extract a minimal triangular set. This special triangular set, defined formally as the W-characteristic set of the ideal, possesses remarkable properties and plays a key role in our work on \src~pairs.

\begin{definition}[{{\cite[Def.~3.1]{w2016o}}}]\label{def:wchar}\rm
	Let $\pset{F}$ be a polynomial set in $\kx$, $\pset{G}$ be the \grobner basis of $\bases{\pset{F}}$, $\Theta_0=\pset{G}\cap\field{K}$, $\pset{G}^{\bases{i}}=\{G\in \pset{G}:\, \lv(G) = x_i\}$,  $\Theta_{i}$ be the set consisting of the smallest polynomial in $\pset{G}^{\bases{i}}$ if $\pset{G}^{\bases{i}}\neq \emptyset$, or $\emptyset$ otherwise ($1 \leq i\leq n$).  The set $\Theta_0\cup\Theta_1\cup \cdots \cup \Theta_n$ of polynomials, ordered with increasing leading variables when $\pset{G}\neq\{1\}$, is called the \emph{W-characteristic set} of $\bases{\pset{F}}$.
\end{definition}

\begin{proposition}
	Let $\ideal{I}$ be any ideal in $\kx$ with $\ideal{I}^{(0)}=\ideal{I}\cap\field{K}=\emptyset$, $\ideal{I}^{(i)}=\ideal{I}\cap\field{K}[x_1,\ldots,x_i]$, $\ideal{I}^{\bases{i}}=\{F\in \ideal{I}^{(i)}\setminus \ideal{I}^{(i-1)}:F \text{ is monic and B-reduced with respect to all }P\in\ideal{I}^{(i-1)}\}$, and $\Omega_{i}$ be the set consisting of the smallest polynomial in $\ideal{I}^{\bases{i}}$ with respect to the lex term order if $\ideal{I}^{\bases{i}}\neq \emptyset$, or $\emptyset$ otherwise $(1 \leq i\leq n)$. Then the set $\Omega_1\cup \cdots \cup \Omega_n$ of polynomials, ordered with increasing leading variables, is the W-characteristic set of $\ideal{I}$.
\end{proposition}

\begin{proof}
Let $\pset{G}$ be the reduced lexicographical \grobner basis of $\ideal{I}$. 
	We need to show that $\Omega_1\cup \cdots \cup \Omega_n$ defined in the proposition is the minimal triangular set contained in $\pset{G}$. Note that $\pset{G}^{\bases{i}}=\pset{G}^{(i)}\setminus\pset{G}^{(i-1)}\subseteq \ideal{I}^{\bases{i}}$ for $1 \leq i\leq n$,  where $\pset{G}^{(0)}=\emptyset$ and
	$\pset{G}^{(i)}=\pset{G}\cap\field{K}[x_1,\ldots,x_i]$.
	For each $i$, if $\Omega_i=\emptyset$, then $\ideal{I}^{\bases{i}}=\emptyset$, so $\pset{G}^{\bases{i}}=\emptyset$.
	
	Now suppose that $\Omega_i\neq\emptyset$ for some $i$ and let $\Omega_i=\{F_i\}$. Then $F_i$ is B-reduced with respect to all the polynomials in $\ideal{I}^{(i-1)}\supseteq \pset{G}^{(i-1)}$. If $\pset{G}^{\bases{i}}=\emptyset$, then $F_i\neq 0$ is B-reduced with respect to $\pset{G}^{(i)}$. This leads to a contradiction because $F_i\in\ideal{I}^{(i)}$.
	Hence $\pset{G}^{\bases{i}}\neq\emptyset$.
	Let  $G_i$ be the smallest polynomial in $\pset{G}^{\bases{i}}$. Then
	$\lt(F_i)\leq\lt(G_i)$ with respect to the lexicographical term order.
	If $\lt(F_i)<\lt(G_i)$, then $F_i$ is B-reduced with respect to $\pset{G}^{(i)}$. This contradicts with the fact that $F_i\in\ideal{I}^{(i)}$. Therefore $\lt(F_i)=\lt(G_i)>\lt(F_i-G_i)$, for $F_i$ and $G_i$ are monic. Since $F_i$ is the smallest polynomial in $\ideal{I}^{\bases{i}}$, $F_i-G_i\in\ideal{I}^{(i-1)}$. As $F_i$ and $G_i$ are both  B-reduced with respect to $\pset{G}^{(i-1)}$, so is $F_i-G_i$. It follows that $F_i-G_i\equiv 0$, and thus $F_i=G_i$.
\end{proof}

The ordered set $\Omega_1\cup \cdots \cup \Omega_n$ in the above proposition equals the W-characteristic set of $\ideal{I}$, so it can be taken as an equivalent definition for the W-characteristic set of any nontrivial polynomial ideal. This alternative definition, which is independent of \grobner basis and is worked out in response to a comment of Mingsheng Wang, does not indicate how to construct the W-characteristic set of the ideal.

\begin{proposition}[{{\cite[Prop.\ 3.1]{w2016o}}}]\label{prop:zero}
	Let $\pset{F}$ be a polynomial set in $\kx$ and $\pset{C}$ be the W-characteristic set of $\bases{\pset{F}} \subseteq \kx$. Then:
	\begin{enumerate}
		\item[$(a)$] for any $F \in \bases{\pset{F}}$, $\prem(F, \pset{C}) = 0$;
		\item[$(b)$] $\bases{\pset{C}} \subseteq \bases{\pset{F}} \subseteq \sat(\pset{C})$;
		\item[$(c)$] $\zero(\pset{C} / \ini(\pset{C})) \subseteq \zero(\pset{F}) \subseteq \zero(\pset{C})$.
	\end{enumerate}
\end{proposition}

We say that the \emph{variable ordering condition is satisfied} for a triangular set $\pset{T}$ if all the parameters of $\pset{T}$ are ordered smaller than the leading variables of the polynomials in $\pset{T}$.
\begin{theorem}[{{\cite[Thm.\ 3.9]{w2016o}}}]\label{thm:divisible}
	Let $\pset{C} = [C_1, \ldots, C_r]$ be the W-characteristic set of $\bases{\pset{F}}\subseteq \kx$. If the variable ordering condition is satisfied for  $\pset{C}$ and $\pset{C}$ is not normal, then there exists an integer $k~(1\leq k < r)$ such that $[C_1, \ldots, C_k]$ is normal and $[C_1, \ldots, C_{k+1}]$ is not regular.
\end{theorem}

Some structural properties about pseudo-divisibility among polynomials in the  \grobner bases can be found in \cite{w2016o}. Based on those properties, an effective algorithm for normal triangular decomposition of polynomial sets has been proposed in \cite{WDM20}.

\begin{definition}\label{def:charpair}\rm
	A pair $(\pset{G},\pset{C})$ of polynomial sets in $\kx$ is called a \emph{characteristic pair} if $\pset{G}$ is a \grobner basis and $\pset{C}$ is the W-characteristic set of $\bases{\pset{G}}$.
	We say that the characteristic pair $(\pset{G},\pset{C})$ is \emph{strong} if $\sat(\pset{C})=\bases{\pset{G}}$.
\end{definition}

A characteristic pair $(\pset{G},\pset{C})$ is said to be \emph{regular} or \emph{normal} if $\pset{C}$ is regular or normal, respectively. Let $(\{1\},[1])$ be regarded as a trivial regular/normal characteristic pair. If a characteristic pair is strong and regular, then the \grobner basis $\pset{G}$ in the pair is said to be \emph{strong regular}. By regular or normal characteristic decomposition of a polynomial set $\pset{F}\subseteq \kx$, we mean a finite set of regular or normal characteristic pairs $(\pset{G}_1,\pset{C}_1), \ldots, (\pset{G}_{e}, \pset{C}_e)$ satisfying the ideal relations in \eqref{eq:chardec}.

The expression \eqref{eq:chardec} can be rewritten in terms of the zero sets or varieties as
\begin{equation} \label{eq:idealdec}
\zero(\pset{F}) = \bigcup_{i=1}^e\zero(\pset{G}_i) =  \bigcup_{i=1}^e\zero(\sat(\pset{C}_i)).
\end{equation}
When the regular or normal characteristic pairs $(\pset{G}_i,\pset{C}_i)$ are computed by the algorithms described in \cite{w00c,WDM20}, the zero relation
\begin{equation} \label{eq:NormalDec}
\zero(\pset{F}) = \bigcup_{i=1}^e\zero(\pset{C}_i / \ini(\pset{C}_i))
\end{equation}
also holds.

To compute a desired characteristic decomposition of a polynomial set $\pset{F}$, we can first proceed to decompose $\pset{F}$ into finitely many \grobner bases $\pset{G}_i$ of certain kinds and then form the characteristic pairs $(\pset{G}_i,\pset{C}_i)$ by simply extracting the W-characteristic sets $\pset{C}_i$ from $\pset{G}_i$.

\section{Computing strong regular characteristic pairs}\label{sec:satdec}

The main objective of this section is to show how to construct
from an arbitrary ideal $\ideal{I}$ an \src~pair $(\pset{G},\pset{C})$ such that $\ideal{I}\subseteq \sat(\pset{C})$ by computing \grobner bases. The construction enables us to devise a novel algorithm for decomposing any polynomial set into finitely many \src~pairs.



\begin{definition}\label{def:characterizable}\rm
	A \grobner basis $\pset{G}$ is said to be \emph{characterizable} if $\bases{\pset{G}}=\sat(\pset{C})$, where $\pset{C}$ is the W-characteristic set of $\bases{\pset{G}}$. 
\end{definition}

Obviously, the W-characteristic set $\pset{C}$ extracted from the \grobner basis $\pset{G}$ is unique. The following proposition shows that if $\sat(\pset{C})=\bases{\pset{G}}$, then $\pset{C}$ must be regular and thus $(\pset{G},\pset{C})$ is an \src~pair. In other words, the \grobner basis $\pset{G}$ in any \src~pair $(\pset{G},\pset{C})$ is characterizable. Otherwise, $\sat(\pset{C})\neq \bases{\pset{G}}$; in this case, $\pset{G}$ may not necessarily be determined by $\pset{C}$ (i.e., there may be two \grobner bases $\pset{G}_1$ and $\pset{G}_2$ such that $\bases{\pset{G}_1}$ and $\bases{\pset{G}_2}$ have the same W-characteristic set $\pset{C}$).

\begin{proposition}\label{prop:isregular}
	The W-characteristic set of the ideal generated by any characterizable \grobner basis is regular.
\end{proposition}
\begin{proof}
	Let $\pset{G}$ be any characterizable \grobner basis and $\pset{C}=[C_1,\ldots,C_r]$ be the W-characteristic set of $\bases{\pset{G}}$.
	For any $P\in \sat(\pset{C})$, we have $P\in \bases{\pset{G}}$ since $\bases{\pset{G}}=\sat(\pset{C})$. From Proposition \ref{prop:zero}(a) one can see that $\prem(P,\pset{C})=0$, which means that $\sat(\pset{C})\subseteq \{P\in\kx: \prem(P,\pset{C})=0\}\subseteq \sat(\pset{C})$. Therefore, by \cite[Thm.~6.1]{a99t} $\pset{C}$ is regular.
\end{proof}

\begin{theorem}[Strong Regularization]\label{the:scharpair}
	Let $\pset{C}_1$ be the W-characteristic set of an arbitrary ideal $\ideal{I}\subset \kx$, $\pset{G}_i$ be the \grobner basis of $\sat(\pset{C}_{i-1})$, and $\pset{C}_i$ be the W-characteristic set of $\bases{\pset{G}_i}$ for $i\geq 2$. Then there exists an integer $m\geq 2$ such that $\ideal{I}\subseteq \sat(\pset{C}_1)\subseteq \sat(\pset{C}_m)$ and either $\pset{C}_m$ is a regular set, or $\sat(\pset{C}_m)=\bases{1}$.
\end{theorem}
\begin{proof}	From Proposition \ref{prop:zero}(a) and $\sat(\pset{C}_{i-1})=\bases{\pset{G}_i}\subseteq \sat(\pset{C}_i)$ we know that $\sat(\pset{C}_1)\subseteq \sat(\pset{C}_{2})\subseteq \cdots \subseteq \sat(\pset{C}_i)\subseteq \cdots$ for $i\geq 2$. Then by the Ascending Chain Condition there exists an $m\geq 2$ such that $\sat(\pset{C}_{m-1})=\sat(\pset{C}_{m})=\cdots$. It follows that the \grobner basis $\pset{G}_{m}$ is characterizable unless $\sat(\pset{C}_m)= \bases{1}$. Therefore, one sees from Proposition \ref{prop:isregular} that $\pset{C}_{m}$ is regular or $\sat(\pset{C}_m)=\bases{1}$.
\end{proof}

The process of constructing the chain of W-characteristic sets $\pset{C}_i$ of $\ideal{I}_{i}=\bases{\pset{G}_{i}}$ for $i=1,\ldots,m$ such that $\pset{C}_m$ is regular and $\sat(\pset{C}_m)=\sat(\pset{C}_{m-1})$ (called \emph{strong regularization}) shown in Theorem~\ref{the:scharpair} is depicted by the diagram in Figure \ref{fig:sregularization}.
\begin{figure}[H]
	\centering
	
	\begin{equation*}
	\begin{array}{cccccccccl}
	\ideal{I} & =& \ideal{I}_1&=&\bases{\pset{G}_1}& \longrightarrow &\pset{C}_1&\longrightarrow & \sat(\pset{C}_1)&(\neq \bases{1})\\
	
	&& \cap\rule{0.5pt}{5pt} && && &&\cap\rule{0.5pt}{5pt} &\\	
	& =& \ideal{I}_2&=&\bases{\pset{G}_2}& \longrightarrow &\pset{C}_2&\longrightarrow & \sat(\pset{C}_2)&(\neq \bases{1})\\
	&& \cap\rule{0.5pt}{5pt} & && &&&\cap\rule{0.5pt}{5pt} &\\[-4pt]
	&& \vdots && & &&&\vdots &\\	
	
	&& \cap\rule{0.5pt}{5pt} && && &&\cap\rule{0.5pt}{5pt} &\\
	& =& \ideal{I}_{m-1}&=&\bases{\pset{G}_{m-1}}& \longrightarrow &\pset{C}_{m-1}&\longrightarrow & \sat(\pset{C}_{m-1})\!\!\!&(\neq \bases{1})\\
	&& \cap\rule{0.5pt}{5pt} & & &&&&\shortparallel &\\[-11pt]
	&&  & & &&&&\shortparallel &\\
	& =& \ideal{I}_m&=&\bases{\pset{G}_m}& \longrightarrow &\pset{C}_m&\longrightarrow & \sat(\pset{C}_m)&=\,\sat(\pset{C}_{m+1})\,=\,\cdots
	
	\end{array}
	\end{equation*}
	
	\caption{Strong regularization of W-characteristic sets by means of saturation}\label{fig:sregularization}
\end{figure}


\begin{remarks}\rm
	From \cite[Thm.\ 4.4]{MW19} we know that for any normal set $\pset{T}$, if its parameters are ordered smaller than the other variables, then the W-characteristic set $\pset{C}$ of $\sat(\pset{T})$ is also normal. This corresponds in some way to the case of Theorem $\ref{the:scharpair}$ when the variable ordering condition is satisfied and $\pset{C}_1$ is assumed to be regular; in this case $\pset{C}_2$ is always regular and thus $m=2$. However, under the hypothesis of Theorem $\ref{the:scharpair}$, $\pset{C}_2$ is not necessarily regular (i.e., $m$ may be greater than $2$), as shown by Example \ref{ex:scharpair} (where $m=3$). Example \ref{ex:counterexample} shows that Theorem 4.4 in \cite{MW19} does not necessarily hold for regular set $\pset{T}$ without the variable ordering condition.
\end{remarks}

For any triangular set $\pset{T}$, $\sat(\pset{T})$ is said to be \emph{equiprojectable} if there exists a regular set $\bar{\pset{T}}$ such that $\sat(\pset{T})=\sat(\bar{\pset{T}})$ \cite{AUB00}. The following corollary follows directly from \cite[Thm.\ 4.4]{MW19}.

\begin{corollary}\label{prop:prop}
	For any triangular set $\pset{T}$ satisfying the variable ordering condition, $\sat(\pset{T})$ is equiprojectable if and only if the W-characteristic set $\pset{C}$ of $\sat(\pset{T})$ is regular and $\sat(\pset{T})=\sat(\pset{C})$.
\end{corollary}


Corollary \ref{prop:prop} points out explicitly how to construct $\bar{\pset{T}}$ from $\pset{T}$ and thus how to check whether $\sat(\pset{T})$ is equiprojectable.

\begin{example}\label{ex:scharpair}\rm
	Let $\pset{C}_1=[x^2-x,(y^2-x)(y-1),(y-1)z]\subseteq \mathbb{K}[x,y,z]$ with $x<y<z$. It is easy to verify that $\pset{C}_1$ is a \grobner basis, it is a triangular set, and it is also the W-characteristic set of $\ideal{I}=\bases{\pset{C}_1}$, but $\pset{C}_1$ is not regular. The  \grobner basis of the saturated ideal of $\pset{C}_1$ may be easily computed as $\pset{G}_2=\{z,x(x-1),y^2-x,x(y+1)\}$. The W-characteristic set of $\bases{\pset{G}_2}=\sat(\pset{C}_1)$ is $\pset{C}_2=[x(x-1),x(y+1),z]$. Obviously $\pset{C}_2$ is not regular and thus $\sat(\pset{C}_1)$ is not equiprojectable, while the W-characteristic set $\pset{C}_3=[x-1,y+1,z]$ of $\sat(\pset{C}_2)$ is regular. Therefore, for this example the integer $m$ in Theorem \ref{the:scharpair} is equal to $3$.
\end{example}

\begin{example}\label{ex:counterexample}\rm
	Consider the regular set $\pset{T}=[y^2,x^2z+xy]$ in $\mathbb{K}[x,y,z]$ with $y<x<z$ (cf.\ \cite[Ex.\,3.20]{WDM20}). The reduced lex \grobner basis $\pset{G}$ of the saturated ideal $\sat(\pset{T})$ of $\pset{T}$ is $\{y^2,yz,xz+y,z^2\}$, so the W-characteristic set $\pset{C}$ of $\bases{\pset{G}}$ is $[y^2,yz]$. One can easily see that $\pset{C}$ is not regular.
\end{example}

Algorithm \ref{alg:strongpair} can be used to compute an \src~pair from any polynomial ideal. Its termination and correctness follow directly from Theorem~\ref{the:scharpair}.

\begin{algorithm}[!h]
	\label{alg:strongpair}
	\KwIn{$\ideal{I}$, an ideal in $\kx$.}
	\KwOut{$(\pset{G},\pset{C})$, an \src~pair such that $\ideal{I}\subseteq \sat(\pset{C})=\bases{\pset{G}}$, or $(\{1\},[1])$ when $\ideal{I}=\bases{1}$.}
	
	{
		\BlankLine		
		$\pset{G}:=$ \grobner basis of $\ideal{I}$;\\
		\eIf{$\pset{G}\neq \{1\}$}
		{
			$\pset{C}:=$ W-characteristic set of $\bases{\pset{G}}$;\\
			\While{$\sat(\pset{C})\neq \bases{\pset{G}}$\label{line:ideal}}
			{	
				$\pset{G} :=$ \grobner basis of $\sat(\pset{C})$;\\
				\If{$\pset{G}=\{1\}$}
				{
					\Return{$(\{1\},[1])$};
				}
				$\pset{C}:=$  W-characteristic set of $\bases{\pset{G}}$;\label{line:wsett}\\
			}
		}
		{
			$\pset{C}:=[1]$;
		}		
		\Return{$(\pset{G},\pset{C})$}
	}\label{line:return}
	\caption{$(\pset{G},\pset{C}):= \algt(\ideal{I})$ (for computing an \src~pair from a polynomial~ideal)}
\end{algorithm}
\setlength{\textfloatsep}{2pt}

\section{Strong regular characteristic decomposition}\label{sec:algsatdec}
In this section, we present an algorithm to compute strong regular characteristic decompositions of polynomial sets using \grobner basis and ideal computations. The main ingredients of the algorithm are two subalgorithms: one for computing \src~pairs with ideal saturation and the other for computing \src~divisors with ideal quotient.


For any two ideals $\ideal{I}$ and $\ideal{J}$, we say that $\ideal{J}$ \emph{divides} $\ideal{I}$ if $\ideal{I}:\ideal{J}\neq \ideal{I}$.

\begin{definition}\rm
	Let $\ideal{I}$ be an ideal in $\kx$ and $\pset{C}$ be the W-characteristic set of $\ideal{I}$. $\pset{C}$ is said to be \emph{strong} if $\sat(\pset{C})=\ideal{I}$, or \emph{morbid} if $\sat(\pset{C})$ does not divide $\ideal{I}$. 
\end{definition}

Let $\pset{C}$ be the W-characteristic set of an ideal $\ideal{I}\subseteq\kx$. If $\pset{C}$ is strong, then $\pset{C}$ carries all the information of $\ideal{I}$. If $\pset{C}$ is morbid, then the associated primes of $\ideal{I}$ are all properly contained in those of $\sat(\pset{C})$; in this case the structure of $\ideal{I}$ is so complicated that $\pset{C}$ carries almost no information of $\ideal{I}$.

\begin{proposition}\label{prop:quot}
	Let $\ideal{I}$ and $\ideal{J}$ be two ideals in $\kx$ such that $\ideal{I}\subseteq  \ideal{J}$. Then $\rad{\ideal{I}}=\rad{\ideal{J}}\cap \rad{\ideal{I}:\ideal{J}}$.
\end{proposition}
\begin{proof}
	On one hand, for any polynomial $P\in \ideal{I}$, we have $P\in \ideal{J}$ and $P\in \ideal{I}:\ideal{J}$ since $\ideal{I}\subseteq \ideal{J}$ and $\ideal{I}\subseteq \ideal{I}:\ideal{J}$. It follows that $P\in \ideal{J}\cap (\ideal{I}:\ideal{J})$, which implies that $\rad{\ideal{I}}\subseteq \rad{\ideal{J}}\cap \rad{\ideal{I}:\ideal{J}}$. On the other hand, for any polynomial $P\in \rad{\ideal{J}}\cap \rad{\ideal{I}:\ideal{J}}$, we have $P\in \rad{\ideal{J}}$ and $P\in \rad{\ideal{I}:\ideal{J}}$. As $P\in \rad{\ideal{J}}$, there exists an $m$ such that $P^m\in \ideal{J}$; as $P$ also belongs to $\rad{\ideal{I}:\ideal{J}}$, one can find an $l$ such that $P^l H\in \ideal{I}$ for any polynomial $H\in\ideal{J}$. Therefore, $P^lP^m\in \ideal{I}$ and thus $P\in \rad{\ideal{I}}$. Hence $\rad{\ideal{J}}\cap \rad{\ideal{I}:\ideal{J}}\subseteq \rad{\ideal{I}}$.
\end{proof}

For two ideals $\ideal{I}$ and $\ideal{J}$ in $\kx$ with $\ideal{I}\subseteq \ideal{J}$, by Proposition \ref{prop:quot}, $\rad{\ideal{I}}=\rad{\ideal{J}}\cap \rad{\ideal{I}:\ideal{J}}$. This decomposition is trivial if $\ideal{I}:\ideal{J}= \ideal{I}$. For any given ideal $\ideal{I}\subsetneq \kx$, we will show in the following subsection how to construct an ideal $\ideal{J}\subseteq\kx$ such that $\ideal{I}\subseteq\ideal{J}$ and $\ideal{I}:\ideal{J}\neq\ideal{I}$.
With the ideal $\ideal{J}$ specially constructed from $\ideal{I}$, the splitting of $\ideal{I}$ to $\ideal{J}$ and $\ideal{I}:\ideal{J}$ is nontrivial and
$\rad{\ideal{J}}$ and $\rad{\ideal{I}:\ideal{J}}$ can be further decomposed.
The technique of splitting ideals by taking quotient according to Proposition \ref{prop:quot} originates from the idea of dividing a known subvariety out of any given variety explained in \cite[pp.\,196--197]{W2001E}.

\subsection{Finding saturated ideals of regular sets that divide a given ideal}

\begin{definition}\label{def:srcdivisor}\rm
	An \src~pair $(\pset{G},\pset{C})$ of polynomial sets is called an \emph{\src~divisor} of a polynomial ideal $\ideal{I}$ in $\kx$ if $\bases{\pset{G}}=\sat(\pset{C})$ divides $\ideal{I}$.
\end{definition}

Given an ideal $\ideal{I}\subsetneq \kx$, we want to find an \src~divisor $(\pset{G}, \pset{C})$ of $\ideal{I}$.
Suppose that $\ideal{I}\neq \bases{1}$ and compute an \src~pair  $(\bar{\pset{G}},\bar{\pset{C}}):=\algt(\ideal{I})$ (using Algorithm~\ref{alg:strongpair}) such that $\ideal{I}\subseteq\sat(\bar{\pset{C}})$. If $\ideal{I}:\sat(\bar{\pset{C}})\neq\ideal{I}$, then $(\pset{G}, \pset{C})=(\bar{\pset{G}},\bar{\pset{C}})\neq (\{1\},[1])$ is
an \src~divisor of $\ideal{I}$.

Otherwise, $\ideal{I}:\sat(\bar{\pset{C}})=\ideal{I}$.
 Let $\pset{G}^*$ be the \grobner basis of $\ideal{I}$ and extract the W-characteristic set $\pset{C}^*$ of $\ideal{I}$ from $\pset{G}^*$.
 Recall that $\ideal{I}\subseteq\sat(\pset{C}^*)$ and $\ideal{I}\neq \bases{1}$.
 If $\sat(\pset{C}^*)=\ideal{I}$, then $\sat(\bar{\pset{C}})=\sat(\pset{C}^*)$ according to Theorem~\ref{the:scharpair} (or Fig.\,1), so $\ideal{I}:\sat(\bar{\pset{C}})=\sat(\pset{C}^*):\sat(\pset{C}^*)=
 \bases{1}\neq\ideal{I}$, which contradicts the assumption that $\ideal{I}:\sat(\bar{\pset{C}})=\ideal{I}$.
 Therefore, $\sat(\pset{C}^*)\neq \ideal{I}$ and one can find a polynomial $F\in \sat(\pset{C}^*)$, $F\notin \ideal{I}$. Let $J=\prod_{C\in \pset{C}^*}\ini(C)$. Then there exists a positive integer $q$ such that
 $FJ^q\in\bases{\pset{C}^*}\subseteq\ideal{I}$. Now choose polynomials $H_1, \ldots, H_t\in\kx\setminus\ideal{I}$ such that\footnote{Natural choices for $\{H_1, \ldots, H_t\}$ in line \ref{line:factors} of Algorithm 2 are $\{F,\, \ini(C): C\in\pset{C}^*,\, \ini(C)\not\in\fk\}$ and $\{H\in\kx: H\mid F\prod_{C\in \pset{C}^*} \ini(C), H~\mbox{is irreducible (or squarefree)}\}$. Our experiments are done with the latter choice for squarefree $H$.} $H_1\cdots H_t=0\iff FJ^q=0$.
	\begin{enumerate}
		\item[(1)]
		Let $(\tilde{\pset{G}},\tilde{\pset{C}})\neq (\{1\},[1])$ be an \src~pair computed by Algorithm $\algt$ from $\bases{\ideal{I}\cup \{H_i\}}$ for some $1\leq i\leq t$. If $\ideal{I}:\sat(\tilde{\pset{C}})\neq \ideal{I}$, then $(\pset{G}, \pset{C})=(\tilde{\pset{G}},\tilde{\pset{C}})$ is an \src~divisor of $\ideal{I}$.
		\item[(2)] Otherwise, $\sat(\tilde{\pset{C}})$ does not divide $\ideal{I}$, where $(\tilde{\pset{G}},\tilde{\pset{C}})$ is computed from $\bases{\ideal{I}\cup \{H_i\}}$ as in (1), for all $1\leq i\leq t$.
In this case, proceed further to find an \src~divisor of $\ideal{I}$ from $\bases{\ideal{I}\cup\{H_i,H_{ij}\}}$ similarly, where $H_{ij}\in\kx\setminus\bases{\ideal{I}\cup \{H_i\}}$ such that $H_{i1}\cdots H_{it_i}=0\iff F_i\prod_{C\in \pset{C}^*_i} \ini(C)=0$, $F_i\in \sat(\pset{C}^*_i)\setminus\bases{\ideal{I}\cup \{H_i\}}$, and $\pset{C}^*_i$ is the W-characteristic set of $\bases{\ideal{I}\cup\{H_i\}}$.
 \end{enumerate}

The above process can continue recursively until an \src~divisor of  $\ideal{I}$ is found. We formulate the process as Algorithm \ref{alg:dividesrc}.
\begin{algorithm}[!h]
	\label{alg:dividesrc}
	\KwIn{$\ideal{I}\neq \bases{1}$, an ideal in $\kx$.}
	\KwOut{$(\pset{G},\pset{C})$, an \src~divisor of $\ideal{I}$.}
	
	{
		\BlankLine
		$(\pset{G}, \pset{C}) :=\algt(\ideal{I})$\label{line:src1};\\
		$\bar{\ideal{I}}:=\ideal{I}:\sat(\pset{C})$;\\
		\If{$\bar{\ideal{I}}=\ideal{I}$}
		{
	        $\pset{G}^*:=$ \grobner basis of $\ideal{I}$;\\
	        $\pset{C}^*:=$ W-characteristic set of $\bases{\pset{G}^*}$;\\
			Choose $F\in \sat(\pset{C}^*)\setminus \ideal{I}$ and
			$H_1, \ldots, H_t\in\kx\setminus\ideal{I}$ such that\\ \qquad $H_1\cdots H_t=0\iff F\prod_{C\in \pset{C}^*} \ini(C)=0$\label{line:factors};\\
			\For{$i=1,\ldots,t$}
			{
				$(\pset{G},\pset{C}):=\algt(\bases{\ideal{I}\cup \{H_i\}})$;\label{line:src2}\\
				$\bar{\ideal{I}}:=\ideal{I}:\sat(\pset{C})$;\\
				\If{$\bar{\ideal{I}}\neq\ideal{I}$}
				{
					\Return{$(\pset{G},\pset{C})$;}\label{return1}
				}
			}
			\For{$i=1,\ldots,t$}
			{ $(\pset{G}, \pset{C}):=\algnot(\bases{\ideal{I}\cup \{H_i\}})$;\\
				$\bar{\ideal{I}}:= \ideal{I}:\sat(\pset{C})$;\\
				\If{$\bar{\ideal{I}}\neq\ideal{I}$}
				{\Return{$(\pset{G}, \pset{C})$\label{return2}};}
			}
		}
		
	}\Return{$(\pset{G}, \pset{C})$}\label{return3}
	\caption{$(\pset{G}, \pset{C}):= \algnot(\ideal{I})$ (for finding an \src~divisor of a polynomial ideal)}
\end{algorithm}
\setlength{\textfloatsep}{2pt}

\begin{proof}[Proof of Termination and Correctness of Algorithm \ref{alg:dividesrc}]
	The process of searching for an \src~divisor of
	the ideal $\ideal{I}$ in Algorithm~\ref{alg:dividesrc} can be viewed as of building up a multi-branch tree $\Gamma$ with $\ideal{I}$ associated to its root, $\ideal{I}_{j_1}$ associated to its nodes of the first level,  $\ideal{I}_{j_1j_2}$ associated to its nodes of the second level, etc. In general, the nodes $j_1 \cdots j_l$ of the $l$th level of $\Gamma$ are associated with $\ideal{I}_{j_1 \cdots j_l}$. For each index $j_1 \cdots j_l$ there is an~\src~pair
	$(\pset{G}_{j_1 \cdots j_l},\pset{C}_{j_1 \cdots j_l})$ such that either (1) $\ideal{I}_{j_1 \ldots j_l}=\bases{1}$, or (2) $\ideal{I}:\sat(\pset{C}_{j_1 \cdots j_l})\neq \ideal{I}$, or (3) $\ideal{I}:\sat(\pset{C}_{j_1 \cdots j_l})= \ideal{I}$. In case (1) or (2), the tree $\Gamma$ terminates to grow and the node $j_1 \cdots j_l$ becomes a leaf. When (2) happens, an \src~divisor of $\ideal{I}$ is found. In case (3), new nodes $j_1 \cdots j_li$ of the $(l+1)$th level are generated by some factors $H_{j_1 \cdots j_li}\not\in \ideal{I}_{j_1 \cdots j_l}$ of $F[\prod_{C\in \pset{C}^*_{j_1 \cdots j_l}} \ini(C)]^q\in\bases{\pset{C}^*_{j_1 \cdots j_l}}\subseteq \ideal{I}_{j_1 \cdots j_l}$, where $F\in \sat(\pset{C}^*_{j_1 \cdots j_l})\setminus\ideal{I}_{j_1 \cdots j_l}$, $q$ is a positive integer, and $\pset{C}^*_{j_1 \cdots j_l}$ is the W-characteristic set of $\ideal{I}_{j_1 \cdots j_l}$. The new nodes are associated with $\bases{\ideal{I}_{j_1 \cdots j_l}\cup\{H_{j_1 \cdots j_li}\}}$ ($s_{j_1 \cdots j_l}\leq i\leq t_{j_1 \cdots j_l}$).

(Termination) As $H_{j_1 \cdots j_li}\not\in \ideal{I}_{j_1 \cdots j_l}$ and $\ideal{I}_{j_1 \cdots j_li}=\bases{\ideal{I}_{j_1 \cdots j_l}\cup\{H_{j_1 \cdots j_li}\}}$,
the containment $\ideal{I}\subset \ideal{I}_{j_1}\subset\cdots\subset \ideal{I}_{j_1 \cdots j_l}\subset \ideal{I}_{j_1\cdots j_l i}$ is strict according to the Ascending Chain Condition for polynomial ideals. Hence the length of every branch of the tree $\Gamma$ from the root to a leaf is finite, so the algorithm must terminate.

	(Correctness) Recall Proposition~\ref{prop:quot} and observe that $\prod\nolimits_{i=s_{j_1 \cdots j_l}}^{t_{j_1 \cdots j_l}} H_{j_1 \cdots j_li}\in \sqrt{\ideal{I}_{j_1 \cdots j_l}}$. Let $\Delta$ be the set of indices of all the leaves of $\Gamma$. Then, after the growth of the tree $\Gamma$ from a node is completed,	\[\sqrt{\ideal{I}}=\bigcap_{\delta\in\Delta}\sqrt{\ideal{I}_{\delta}}\] always holds. Therefore, at the termination of the algorithm, either an \src~divisor $(\pset{G},\pset{C})$ is found, or $\bases{1}$ is associated to all the leaves of $\Gamma$; the latter contradicts $\ideal{I}\neq\bases{1}$. The \src~divisor produced in line \ref{return1}, \ref{return2}, or \ref{return3} is what we want to find.
\end{proof}

The following theorem follows directly from Algorithm \ref{alg:dividesrc}.

\begin{theorem}\label{the:srcdivisor}
	From any ideal $\ideal{I}\neq \bases{1}$ in $\kx$, one can construct an \src~divisor of $\ideal{I}$.
\end{theorem}

\subsection{An ideal-division-based algorithm for \src~decomposition}
In this subsection, we describe an algorithm to decompose any polynomial set, or the ideal $\ideal{I}$ it generates, into finitely many \src~pairs. Based on the strategy of ideal splitting with quotient, the algorithm works by finding an \src~divisor $(\pset{G}, \pset{C})$ of $\ideal{I}$ and then dividing the saturated ideal of $\pset{C}$ out of $\ideal{I}$ iteratively.

Consider an ideal $\ideal{I}$, initially with input $\pset{F}\subseteq \kx$ as its generating set of polynomials. Let $\Psi$ be the set of \src~pairs already computed; it is $\emptyset$ initially. 
Suppose that $\ideal{I}\neq \bases{1}$ and compute $(\pset{G},\pset{C}):=\algnot(\ideal{I})$;
in this case, $\ideal{I}$ is split into $\sat(\pset{C})$ and $\ideal{I}:\sat(\pset{C})$ and an \src~pair $(\pset{G},\pset{C})\neq (\{1\},[1])$ is obtained and adjoined to $\Psi$. From Proposition~\ref{prop:quot} we know that
	\begin{equation}\label{equ:rad1}
	\rad{\ideal{I}}=\rad{\sat(\pset{C})}\textstyle\,\bigcap \rad{\ideal{I}:\sat(\pset{C})},
	\end{equation}
	so the procedure can continue to decompose  $\ideal{I}:\sat(\pset{C})$ instead of $\ideal{I}$.

The above process of decomposition will terminate in finitely many iterations. We formulate the process as Algorithm~\ref{alg:satdec}. 

\begin{algorithm}[!h]
	\label{alg:satdec}
	\KwIn{$\pset{F}$, a finite, nonempty set of polynomials in $\kx$.}
	\KwOut{$\Psi$, an \src~decomposition of $\pset{F}$ such that $\rad{\bases{\pset{F}}}=\bigcap_{\pset{C} \in\Psi}\rad{\sat(\pset{C})}=\bigcap_{\pset{G}\in \Psi}\rad{\bases{\pset{G}}}$, or the empty set when $\bases{\pset{F}}=\bases{1}$.}
	
	{
		\BlankLine		
		$\Psi := \emptyset$;\\
		$\ideal{I}:=\bases{\pset{F}}$;\\
		\While{$\ideal{I} \neq \bases{1}$\label{line:ideal1}}
		{	
			$(\pset{G},\pset{C}):=\algnot(\ideal{I})$;\\
			$\Psi:=\Psi\cup \{(\pset{G}, \pset{C})\}$\label{line:psi};\\
			$\ideal{I}:=\ideal{I}:\bases{\pset{G}}$\label{line:newp};\\
		}		
		\Return{$\Psi$}
	}\label{line:return2}
	\caption{$\Psi:= \algnorr(\pset{F})$ (for computing an \src~decomposition of a polynomial set)}
\end{algorithm}
\setlength{\textfloatsep}{2pt}

\begin{proof}[Proof of Termination and Correctness of Algorithm \ref{alg:satdec}]
	(Termination) Decomposition in Algo\-ri\-thm \ref{alg:satdec} is an iterative process: every time after the ideal $\ideal{I}$ is decomposed, a new ideal $\bar{\ideal{I}}\supsetneq\ideal{I}$ is obtained and then further decomposed in the same way. Since every new ideal generated in line \ref{line:newp} is strictly enlarged, by the Ascending Chain Condition the algorithm terminates.
	
	(Correctness) When $\bases{\pset{F}}=\bases{1}$ in line \ref{line:ideal1}, $\Psi=\emptyset$ is returned. Hence we only need to show that when $\bases{\pset{F}}\neq\bases{1}$, $\Psi$ is an \src~decomposition of $\pset{F}$, namely all the pairs in $\Psi$ are \src~pairs and the relation \eqref{eq:chardec} holds. It is clear that only in line \ref{line:psi} is adjoined to $\Psi$  a new pair $(\pset{G},\pset{C})$ which is computed by Algorithm \ref{alg:dividesrc}. Therefore, all the pairs in $\Psi$ are \src~pairs.
	
	Now we prove that $\rad{\bases{\pset{F}}}=\bigcap_{\pset{C} \in\Psi}\rad{\sat(\pset{C})}$.
	For those \src~pairs generated in line \ref{line:psi}, one can easily see that $\ideal{I}\subseteq \sat(\pset{C})$; thus from Proposition \ref{prop:quot} we know that the relation \eqref{equ:rad1} holds.
		
	The relation \eqref{equ:rad1} implies that every polynomial $F\in \rad{\ideal{I}}$ is in the intersection of $\rad{\sat(\pset{C})}$ and $\rad{\bar{\ideal{I}}}$, where $\bar{\ideal{I}}$ is a new ideal obtained which remains for further processing. This proves that the relation $\rad{\bases{\pset{F}}}=\bigcap_{\pset{C} \in\Psi}\rad{\sat(\pset{C})}$ holds when the algorithm terminates with $\bar{\ideal{I}}=\bases{1}$. Since $\bases{\pset{G}}=\sat(\pset{C})$ for every \src~pair $(\pset{G},\pset{C})\in \Psi$, $\rad{\bases{\pset{F}}}=\bigcap_{\pset{G}\in \Psi}\rad{\bases{\pset{G}}}$.	This completes the proof of the relation \eqref{eq:chardec}.
\end{proof}

\section{Examples and experiments}\label{sec:exex}
\subsection{Examples for \src~decomposition}
\begin{example}\rm
	Let $\pset{F}=\{uxy, vy^2+y, vx^2+y^2\}\subseteq \mathbb{K}[u,v,x,y]$ with $u<v<x<y$. The \grobner basis of $\ideal{I}_1=\bases{\pset{F}}$ can be easily computed as
	\begin{equation*}
	\begin{aligned}
	\pset{G}_1 = &\{uvx^2, v^4x^4+vx^2, y-v^2x^2\}
	\end{aligned}
	\end{equation*}
	and the W-characteristic set of $\bases{\pset{G}_1}$ is
	$\pset{C}_1=[uvx^2, y-v^2x^2]$.
	The saturated ideal of $\pset{C}_1$ is $\sat(\pset{C}_1)=\bases{x^2,y}$. One sees that 
	$(\bar{\pset{G}}_1,\bar{\pset{C}}_1)$ is an \src~pair, where $\bar{\pset{G}}_1=\{x^2,y\}$ is the \grobner basis and $\bar{\pset{C}}_1=[x^2,y]$ is the W-characteristic set of $\sat(\pset{C}_1)$. It is easy to verify that
	$$\ideal{I}_2=\ideal{I}_1:\sat(\bar{\pset{C}}_1)=\bases{\pset{G}_2}
	$$ is strictly larger than $\ideal{I}_1$, where $\pset{G}_2=\{uv, vP, Q\}$ is the \grobner basis of $\ideal{I}_2$ and $P=v^3x^2+1$, $Q=y-v^2x^2$. Hence $(\bar{\pset{G}}_1,\bar{\pset{C}}_1)$ is adjoined to $\Psi$ and the procedure continues to decompose $\ideal{I}_2$ instead of $\ideal{I}_1$. The W-characteristic set $\pset{C}_2$ of $\ideal{I}_2$ consists of the same polynomials as $\pset{G}_2$:
	$$\quad\pset{C}_2=[C_1, C_2, C_3]=[uv, vP, Q].$$
	Since $\sat(\pset{C}_2)=\bases{1}$, the pair $(\{1\},[1])$ is obtained and $\ideal{I}_2:\bases{1}=\ideal{I}_2$. As $1\in\sat(\pset{C}_2)$, \[\ini(C_2)\cdot\ini(C_1)=v^4u=uC_2-v^3xC_1\in\bases{\pset{C}_2},\] so we can take, e.g., $H_1=v$ and $H_2=u$. Compute the \grobner basis $\pset{G}_3=\{v,y\}$ of $\bases{\pset{G}_2\cup \{H_1\}}$ and extract the W-characteristic set $\pset{C}_3=[v,y]$ of $\bases{\pset{G}_3}$ from $\pset{G}_3$; one finds that $\pset{C}_3$ is regular and $\sat(\pset{C}_3)=\bases{v,y}=\bases{\pset{G}_3}$ and $\ideal{I}_2:\sat(\pset{C}_3)=\bases{u, P, Q}$ is strictly larger than $\ideal{I}_2$. Then the \src~pair $(\pset{G}_3,\pset{C}_3)$ is obtained and adjoined to $\Psi$. The procedure continues with $\ideal{I}_3=\bases{u, P, Q}$.
	
	Simple computation shows that the \grobner basis and the W-characteristic set of $\ideal{I}_3$ contain the same polynomials: $$\pset{G}_4=\{u, P, Q\},\quad \pset{C}_4=[u, P, Q].$$ One can check that $(\pset{G}_4,\pset{C}_4)$ is an \src~pair and $\ideal{I}_3:\sat(\pset{C}_4)=\bases{1}$. Therefore, $(\pset{G}_4, \pset{C}_4)$ is added to $\Psi$ and the procedure terminates. Finally, an \src~decomposition $$\Psi=\{(\bar{\pset{G}}_1,\bar{\pset{C}}_1),(\pset{G}_3,\pset{C}_3),(\pset{G}_4,\pset{C}_4)\}$$ of $\pset{F}$ is obtained.
\end{example}

\begin{example}\rm
	Consider the polynomial set
	$$\pset{F}=\{-ct^2u+t^3-uv^2-uw^2, -ct^2v+t^3-u^2v-vw^2, -ct^2w+t^3-u^2w-v^2w\}$$ (which is Ex 9 in Table \ref{tab:main}) with variable ordering $w<v<u<t<c$. The ideal generated by $\pset{F}$, which is of dimension 2 and not radical, consists of $8$ primary components (none of them is embedded). The polynomial set $\pset{F}$ can be decomposed by Algorithm \ref{alg:satdec} into 6 \src~pairs
$$\begin{array}{l}\smallskip
(\pset{G}_1,\pset{C}_1)=(\{v,u,t\},[v,u,t]),~
(\pset{G}_2,\pset{C}_2)=(\{w,u,t^2\},[w,u,t^2]),~
(\pset{G}_3,\pset{C}_3)=(\{w,v,t^2\},[w,v,t^2]),\\ \smallskip
(\pset{G}_4,\pset{C}_4)=(\{v-w,u-w,wt^2c-t^3+2w^3\},[v-w,u-w,wt^2c-t^3+2w^3]),\\ \smallskip
(\pset{G}_5,\pset{C}_5)=(\{v-w,P,Q,t^2c-wu+w^2\},[v-w,P,Q]),\\
(\pset{G}_6,\pset{C}_6)=(\{G_1,\ldots,G_5\},[G_1,G_2,G_3]),
\end{array}$$
where the polynomials $G_1,\ldots,G_5$ consist of 6, 3, 4, 10, 3 terms respectively and
\[P=t^3-wu^2-w^2u,\quad Q=u^2c+wuc-ut-wt.\]
It may be observed that (1) $\pset{C}_i$ is normal for all $i$, (2) $\rad{\bases{\pset{G}_i}}$ is prime for $i=1,\ldots,5$, and (3) $\rad{\bases{\pset{G}_6}}$ is composed of 3 prime ideals.
\end{example}

The cyclic-$n$ systems are well-known examples for which triangular decompositions based on pseudo-division are more difficult to compute than \grobner bases. Our algorithm can compute \src~pairs for cyclic-6 in a few minutes, while other triangular decomposition algorithms cannot (see Ex 21 in Table \ref{tab:main}).

\subsection{Implementation and experimental results}\label{sec:exp}
We have implemented Algorithm \ref{alg:satdec} as a {\sc Maple} function {\sf srcDec} and carried out experiments with the implementation on an Intel(R) Core(TM) i5-4210U CPU at 1.70 GHz$\times 4$ with 7.7 GB RAM under Ubuntu 16.04 LTS. The implementation is based on the functions for \grobner basis computation available in the {\sf FGb} library and {\sc Maple}'s built-in packages. Selected results of experiments on some test examples are presented in Table~\ref{tab:main}: Ex 1--4 are taken from the Epsilon library, Ex 5--9 from \cite{M2012d}, Ex 10 from \cite{c07c}, Ex 11 from the {\sf FGb} library, and Ex 12--16 can be found at \url{http://www.lifl.fr/~lemaire/BCLM09/BCLM09-systems.txt}, Ex 17--21 at \url{http://homepages.math.uic.edu/~jan/Demo/TITLES.html}, and Ex 22--27 at \url{ http://www-sop.inria.fr/saga/POL/BASE/3.posso}. The function {\sf srcDec} is implemented for direct decomposition of polynomial sets into \src~pairs.
To observe the performance of our algorithm, we made comparative experiments on {\sf srcDec}  in {\sc Maple} 18 with two other relevant functions for unmixed decomposition of polynomial sets into \grobner bases of saturated ideals of triangular sets: one is the {\sf Epsilon} function {\sf uvd} which is implemented for decomposing an arbitrary algebraic variety into unmixed subvarieties. The other function computes first the decomposition of a polynomial set into regular sets using the {\sf RegularChains} function {\sf Triangularize} and then the \grobner bases of the saturated ideals of the computed regular sets.




In Table \ref{tab:main}, Label indicates the label used in the above-cited references and Var, Pol, and Dim denote the number of variables, the number of polynomials in the example, and the dimension of the ideal generated by the polynomials, respectively. Total, GB, SAT, and QUO under Algorithm~\ref{alg:satdec} record respectively the total time (followed by the number of \src~pairs in parentheses) for \src~decomposition using Algorithm \ref{alg:satdec}, the time for computing all the \grobner bases, saturated ideals, and ideal quotients; Total under {\sf uvd} records the total time for unmixed decomposition. Total and Regular
under {\sf Triangularize} record the total time for unmixed decomposition (followed by the number of components in parentheses) and the time for regular decomposition respectively.

\setlength{\dbltextfloatsep}{5pt}
\begin{table*}[h]
	\centering
	\caption{Timings for \src~decomposition (in second)}\label{tab:main}
	{\footnotesize
		\begin{tabular}{cc|c@{\ }c@{\ }c|c@{\  }c@{\  }c@{\  }c|c|c@{\ }c}
			\multicolumn{5}{c}{ } & \multicolumn{4}{c}{Algorithm \ref{alg:satdec}}& \multicolumn{1}{c}{\sf uvd}&\multicolumn{2}{c}{\sf Triangularize}\\
			\hline
			Ex & Label & Var & Pol & Dim   & Total  & GB & SAT & QUO  & Total & Total  & Regular\\
			\hline
			1&E1  &10&10&1   &1.450(2)&0.564&0.221&0.632   &0.867(3)   &2.041(13)&1.709\\
			2&E5  &15&17&4   &3.568(1)&1.793&0.461&1.255   &3.374(1)
			&7.684(7)   &5.225\\
			3&E11  &4&3&1   &0.037(1)&0.008&0.008&0.015   &0.122(1)   &0.079(3)&0.035\\
			4&E34  &14&16&0   &0.574(0)&0.556&0.&0.  &$>1000$ &50.650(0)&50.649 \\
			5&S1  &4&3&2    &0.055(2)&0.017&0.007&0.022   &0.011(3)    &0.054(3)     &0.034   \\
			6&S6  & 4 & 3  &2 &0.175(3)  &0.035&0.018&0.111     &0.032(3)    &0.080(3)    &0.052  \\
			7&S7  & 4 & 3  &1 &0.098(1)  &0.008 &0.052&0.031      &0.652(1)   &0.127(1)    &0.084  \\
			8&S8 & 4 & 3  &2 &0.058(2)  &0.018 &0.004&0.030     &0.112(2)    &0.099(2)  &0.081  \\
			9&S14  & 5 & 3  &2 & 1.024(6)  &0.372&0.299 &0.235     &3.767(6)    &0.184(8)  &0.109  \\
			10&nueral   &4&3&1  &1.225(2)&0.445&0.187&0.574     &$>1000$    & 0.222(5) &0.115        \\
			11&F663    & 10 & 9 & 2 &8.982(1)&1.264  &1.684&5.983 &$>1000$
			&4.312(4)  &0.935  \\
			12&Katsura-4&5&5&0 &0.353(1)&0.254&0.&0.086&209.831(1)&6.583(4)&5.343\\
			13&nld-3-4&4&4&0 &2.037(8)&0.175&0.624&1.121&1.066(8)&0.754(29)&0.520\\
			14&nld-8-3&3&3&0&2.464(2)&2.254&0.207&0.&3.334(2)&17.360(8)&0.268\\       15&nql-5-4&5&5&0&2.908(1)&2.088&0.&0.806&25.422(1)&35.508(1)&30.335\\  16&caprasse-li&4&4&0&104.055(2)&0.066&101.836&2.145&2.883(3)&0.967(7)&0.634\\
			17&Cyclic-5  & 5  & 5  &0 &0.497(4)&0.101&0.058&0.285  &$>1000$  &1.079(15) &0.798  \\
			18&geneig &6&6&0&0.566(1)&0.135&0.003&0.417&$>1000$&--&$>1000$\\
			19&reimer5&5&5&0&32.850(1)&30.983&0.005&1.820&$>1000$&--&$>1000$\\		
			20&redcyc6&6&6&0&36.920(6)&0.203&0.556&35.079&$>1000$&--&$>4000$\\
			21&Cyclic-6  & 6  & 6  &0 &368.988(6)&2.455&332.158&33.780  &$>1000$  &-- &$>1000$  \\
			22&Bronstein2& 4&3&2& 0.377(5)&0.188& 0.087&0.&150.684(4)&0.430(5)&0.387\\
			23&Morgenstern& 9&5&4& 0.100(1)&0.028& 0.014&0.042&2.812(1)&0.223(1)&0.191\\
			24&Raksanyi& 8&4&4& 0.057(1)&0.012& 0.001&0.028&0.588(1)&0.064(2)&0.047\\
			25&Gerdt 5& 13&15&2& 3.150(4)&1.163& 0.330&1.343&2.830(1)&2.009(3)&1.849\\
			26&Czapor Geddes 3& 11&2&9& 1098.118(1)&0.597& 1097.444&0.095&$>4000$&1.261(1)&0.076\\				
			\hline
		\end{tabular}
	}
\end{table*}

The time-consuming steps in algorithm $\algnorr$ are for the computation of the \grobner bases of the input ideals, the \grobner bases of ideal quotients, and the \grobner bases of saturated ideals. The regular sets computed by the algorithm are normal in most cases and they are often simpler than the corresponding regular sets computed by pseudo-division or subresultant-based algorithms, so that the \grobner bases of the saturated ideals of the regular sets produced by our algorithm {\sf srcDec} tend to be easier to compute than those produced by {\sf uvd} or {\sf Triangularize} for about two thirds of the test examples. More importantly, the number of \src~pairs in an \src~decomposition computed by {\sf srcDec} is usually smaller than the number of components in the corresponding unmixed decomposition computed by {\sf uvd} or regular decomposition by {\sf Triangularize}, as shown by the experimental data in Table \ref{tab:main}. This is because our ideal-division-based algorithm generates few redundant components. For example, the computation of the \src~decomposition using {\sf srcDec} takes much less time than that of the regular decomposition using {\sf Triangularize} for Ex 14, as the \src~decomposition contains only two \src~pairs.

\section{Conclusion}\label{sec:con}

In this paper, it is shown that a strong regular characteristic (\src) pair, and therefore an \src~divisor $(\pset{G},\pset{C})$ of $\ideal{I}$, can be constructed from any polynomial ideal $\ideal{I}$ with given generating set $\pset{F}$. The constructed \src~divisor may be used to split the ideal $\ideal{I}$ into the saturated ideal $\sat(\pset{C})$ and the quotient ideal $\ideal{I}:\sat(\pset{C})$. The process of construction and splitting can be repeated for $\ideal{I}:\sat(\pset{C})$ instead of $\ideal{I}$ and recursively, yielding an algorithm capable of decomposing the polynomial set $\pset{F}$ into finitely many \src~pairs $(\pset{G}_1,\pset{C}_1),\ldots,(\pset{G}_e,\pset{C}_e)$ such that
\begin{equation} \label{eq:chardeceq}
\zero(\pset{F}) = \bigcup_{i=1}^e\zero(\pset{G}_i) =  \bigcup_{i=1}^e\zero(\sat(\pset{C}_i))
\end{equation}
or equivalently \eqref{eq:chardec} holds. The relation \eqref{eq:chardeceq} provides two representations for the zero set of $\pset{F}$: one in terms of the
\grobner bases $\pset{G}_1,\ldots,\pset{G}_e$ and the other in terms of the regular sets $\pset{C}_1,\ldots,\pset{C}_e$. Several nice properties about strong regular characteristic pairs and characteristic decompositions have been presented, and the implementation and performance of our proposed algorithm have been illustrated by examples and experimental results.

The contributions of this paper include: (1) two main theorems (Theorems \ref{the:scharpair} and \ref{the:srcdivisor})
 showing
 how to construct an \src~pair and an \src~divisor of an arbitrary polynomial ideal; (2) an algorithm for decomposing any polynomial set $\pset{F}$ into \src~pairs $(\pset{G}_i,\pset{C}_i)$ such that \eqref{eq:chardec} and \eqref{eq:chardeceq} hold; (3) some experiments with a preliminary implementation of the decomposition algorithm.

The triangular sets in an \src~decomposition are \emph{normal} in most cases (cf.\ \cite{WDM20, w2016o}). It turns out that comprehensive triangular decompositions \cite{c07c} and/or \grobner systems \cite{MON02} can be reproduced rather easily from \src~pairs computed by our algorithm, and we are working on the details. The W-characteristic set of an ideal may be morbid and it is not yet clear when morbidity happens. How to establish equivalent conditions for a W-characteristic set to be morbid and how to retrieve information of an ideal from its morbid W-characteristic set are some of the questions that remain for further investigation.

\medskip

The authors wish to thank the referees for their insightful comments which helped bring the paper to the present form. This work also benefited from frequent discussions the authors had with Chenqi Mou.


	\bibliographystyle{abbrv}
	\bibliography{strongdec}

\appendix
\section{ }
\subsection*{A.1. Strong regular characteristic decomposition of radical ideals}

%

	\begin{lemma}\label{lem:common}
		Let $\ideal{I}$ and $\ideal{J}$ be two radical ideals in $\kx$. Then the algebraic varieties $\zero(\ideal{I}:\ideal{J})$ and $\zero(\ideal{J})$ do not have any common irredundant irreducible component.
	\end{lemma}
	\begin{proof}
		Let $\zero(\ideal{I}:\ideal{J})=V_1\cup \cdots \cup V_s$ be an irredundant irreducible decomposition of the algebraic variety $\zero(\ideal{I}:\ideal{J})$.  Suppose that $\zero(\ideal{I}:\ideal{J})$ and $\zero(\ideal{J})$ have a common irredundant irreducible component, say $V_m$ for some positive integer $m$ ($\leq s$), and let
 $$V^*=\bigcup_{\scriptsize\begin{array}{c} 1\leq i\leq s \\  i\neq m\end{array}}V_i.$$
 Then $\zero(\ideal{I})\setminus\zero(\ideal{J})\subseteq V^*\subsetneq V_1\cup \cdots \cup V_s=\zero(\ideal{I}:\ideal{J})$; in this case, $\zero(\ideal{I})\setminus\zero(\ideal{J})\subset V^*$ implies that $\overline{\zero(\ideal{I})\setminus\zero(\ideal{J})}\subseteq \overline{V^*}=V^*$, which contradicts with the fact that $\zero(\ideal{I}:\ideal{J})=\overline{\zero(\ideal{I})\setminus\zero(\ideal{J})}$.
	\end{proof}

\begin{proposition}\label{prop:radical}
	Let $\ideal{I}$ and $\ideal{J}$ be two radical ideals in $\kx$. Then:
	\begin{enumerate}
		\item[$(a)$]  the saturated ideal of the W-characteristic of $\ideal{I}$ is radical;
		\item[$(b)$] the ideal quotient $\ideal{I}:\ideal{J}$ is also radical.
	\end{enumerate}
\end{proposition}
\begin{proof}
	$(a)$ Let $\pset{T}$ be the W-characteristic set of $\ideal{I}$. To show that $\sat(\pset{T})$ is radical, it suffers to show that $\sat(\pset{T})=\rad{\sat(\pset{T})}$. The inclusion $\sat(\pset{T})\subseteq \rad{\sat(\pset{T})}$ is obvious. For any $P\in \rad{\sat(\pset{T})}$, we have $P^lI^m\in \bases{\pset{T}}\subseteq \ideal{I}$ for some nonnegative integers $m$ and $l$, where $I=\prod_{T\in\pset{T}}\ini(T)$. Thus ${(IP)}^{\max(m,l)}\in \ideal{I}$; it follows that $IP\in \rad{\ideal{I}}=\ideal{I}$, so that $\prem(IP, \pset{T})=0$, which means that $P\in \sat(\pset{T})$. Therefore, $\sat(\pset{T})=\rad{\sat(\pset{T})}$.
	
	$(b)$ To prove that $\ideal{I}:\ideal{J}$ is radical, we only need to prove that $\rad{\ideal{I}:\ideal{J}}\subseteq \ideal{I}:\ideal{J}$. For any $F\in \rad{\ideal{I}:\ideal{J}}$ and $J\in \ideal{J}$, there exists a nonnegative integer $l$ such that $F^lJ\in \ideal{I}$; it follows that ${(FJ)}^l\in \ideal{I}$, so that $FJ\in \rad{\ideal{I}}=\ideal{I}$. Hence $F\in \ideal{I}:\ideal{J}$.
\end{proof}


By Proposition \ref{prop:radical}, for any radical ideal $\ideal{I}$ the saturated ideal $\sat(\pset{C})$ of every regular set $\pset{C}$ appearing in an \src~decomposition $\Psi$ of $\ideal{I}$ is radical. According to Lemma \ref{lem:common}, for any two \src~pairs $(\pset{G}_i,\pset{C}_i)$ and $(\pset{G}_j,\pset{C}_j)$ ($i\neq j$) in $\Psi$ the two varieties $\zero(\sat(\pset{C}_i))$ and $\zero(\sat(\pset{C}_j))$ do not have any common irredundant irreducible component. In other words, the ideals $\sat(\pset{C}_i)$ and $\sat(\pset{C}_j)$ have no common minimal associated prime. Therefore, the number of \src~pairs in $\Psi$ must be smaller than that of irredundant irreducible components of the variety $\zero(\pset{F})$. We are unable to establish such a tight bound for nonradical ideals because when $\ideal{I}$ and $\ideal{J}$ in Lemma~\ref{lem:common} are not radical, $\ideal{I}:\ideal{J}$ and $\ideal{J}$ may have a common minimal associated prime. Taking $\ideal{I}=\bases{x^3 y}$ and $\ideal{J}=\bases{x}$ as an example, we have $\ideal{I}:\ideal{J}=\bases{x^2 y}$, which has a minimal associated prime $\bases{x}$ as $\ideal{J}$.

\subsection*{A.2. Triangular decomposition and characteristic decomposition}
Let $\ideal{I}\subset \kx$ be a polynomial ideal with generating set $\pset{F}$ and $\Psi=\{(\pset{G}_1,\pset{C}_1),\ldots,(\pset{G}_e,\pset{C}_e)\}$ be an \src~decomposition of $\pset{F}$ computed by Algorithm \ref{alg:satdec}.
Then $\{\pset{C}_1,\ldots,\pset{C}_e\}$ is a regular triangular decomposition of $\pset{F}$ obtained from $\Psi$ with no further computation. Therefore, the decomposition $\Psi$ provides two representations for $\zero(\pset{F})$: one in terms of the strong regular \grobner bases $\pset{G}_1,\ldots,\pset{G}_e$, and the other in terms of the regular triangular sets $\pset{C}_1,\ldots,\pset{C}_e$.

\begin{theorem}
	Let $\{\pset{T}_1,\ldots, \pset{T}_e\}$ be a regular triangular decomposition of $\pset{F}\subseteq \kx$ and each $(\pset{G}_i,\pset{C}_i)$ be constructed from $\pset{T}_i$ according to Theorem 4.4 in \cite{MW19}. Assume that the variable ordering condition is satisfied for all $\pset{T}_i$. Then $\{(\pset{G}_1,\pset{C}_1),\ldots, (\pset{G}_e,\pset{C}_e)\}$ is a strong normal characteristic decomposition of $\pset{F}$.
\end{theorem}
\begin{proof}
	As $\{\pset{T}_1,\ldots, \pset{T}_e\}$ is a triangular decomposition of $\pset{F}$, we have $\rad{\bases{\pset{F}}}=\bigcap_{i=1}^e \rad{\sat(\pset{T}_i)}$. By Theorem 4.4 in \cite{MW19} the W-characteristic set $\pset{C}_i$ of $\sat(\pset{T}_i)$ is normal and $\sat(\pset{C}_i)=\bases{\pset{G}_i}$, so each $(\pset{G}_i,\pset{C}_i)$ is a strong normal characteristic pair for $i=1,\ldots,e$. Using the equalities $\sat(\pset{T}_i)=\sat(\pset{C}_i)$ and $\bases{\pset{G}_i}=\sat(\pset{T}_i)$, one can easily prove that
	$$\rad{\bases{\pset{F}}}=\bigcap_{i=1}^e \rad{\sat(\pset{T}_i)}=\bigcap_{i=1}^e \rad{\sat(\pset{C}_i)}=\bigcap_{i=1}^e \rad{\bases{\pset{G}_i}}.$$
	Therefore, $\{(\pset{G}_1,\pset{C}_1),\ldots, (\pset{G}_e,\pset{C}_e)\}$ is a strong normal characteristic decomposition of $\pset{F}$.
\end{proof}

More generally, for any decomposition of an ideal $\ideal{I}=\bases{\pset{F}}\subseteq \kx$ into $e$ ideals $\ideal{I}_1,\ldots,\ideal{I}_e$ such that $\rad{\ideal{I}}=\bigcap_{i=1}^{e}\rad{\ideal{I}_i}$, one can apply Algorithm \ref{alg:satdec} to each generating set of $\ideal{I}_i$ for $i=1,\ldots,e$
and then combine the results for all $i$ to yield a strong regular characteristic decomposition of $\pset{F}$, from which a regular triangular decomposition of $\pset{F}$ is obtained as by-product.

\end{document}